\documentclass[onecolumn,12pt]{IEEEtran} 
\usepackage{amssymb}
\usepackage{amsthm,color}
\usepackage{footnote}
\usepackage{latexsym}
\usepackage{graphicx}
\usepackage{fancyhdr}
\usepackage{bbm}
\usepackage{amsmath}
\usepackage{latexsym,amssymb,amsxtra,mathrsfs,times}
\usepackage{cite}
\usepackage{multirow}

\renewcommand{\vec}[1]{\underline{#1}}

\newtheorem{thm}{Theorem}
\newtheorem{theorem}[thm]{Theorem}
\newtheorem{lemma}[thm]{Lemma}

\theoremstyle{definition}

\newcommand{\Z}{\mathbb{Z}}
\newcommand{\calP}{\mathcal{P}}
\newcommand{\F}{\mathcal{F}}

\newcommand{\gf}{{\mathrm{GF}}}
\newcommand{\im}{{\mathrm{Im}}}
\newcommand{\sq}{{\rm SQUARE}}

\def\qi#1 {\fbox {\footnote {\ }}\ \footnotetext { From Qi: {\color{red}#1}}}

\begin{document}

\title{Three New Families of Zero-difference Balanced Functions with Applications}

\author{Cunsheng Ding,\thanks{C. Ding is with the Department of Computer Science and Engineering, The Hong Kong University of Science and Technology, Clear Water Bay, Kowloon, Hong Kong, China (email: cding@ust.hk).}
Qi Wang,\thanks{Q. Wang is with the Department of Computer Science and Engineering, The Hong Kong University of Science and Technology, Clear Water Bay, Kowloon, Hong Kong, China (email: qiwang@ust.hk).}
and Maosheng Xiong \thanks{M. Xiong is with the Department of Mathematics, The Hong Kong University of Science and Technology, Clear Water Bay, Kowloon, Hong Kong (email: mamsxiong@ust.hk).}
}

\maketitle


\renewcommand{\thefootnote}{}




\begin{abstract}
Zero-difference balanced (ZDB) functions integrate a number of subjects in combinatorics and algebra, and have many applications
in coding theory, cryptography and communications engineering. In this paper, three new families of ZDB functions are presented. The first construction, inspired by the recent work \cite{Cai13}, gives ZDB functions defined on the abelian groups $(\gf(q_1) \times \cdots \times \gf(q_k), +)$ with new and flexible parameters. The other two constructions are based on $2$-cyclotomic cosets and yield ZDB functions on $\Z_n$ with new parameters. The parameters of optimal constant composition codes, optimal and perfect difference systems of sets obtained from these new families of ZDB functions are also summarized.
\end{abstract}


\section{Introduction}\label{sec-into}

Let $(A,+)$ and $(B,+)$ be two abelian groups with orders $n$ and $\ell$ respectively. A function $f$ from $A$ to $B$ is called {\em zero-difference balanced} (ZDB for short) if
\[|\{x \in A: f(x+a)-f(x)=0\}|=\lambda, \]
for every nonzero $a \in A$, where $\lambda$ is a non-negative integer. Let $\im(f) = \{b_0, b_1, \ldots, b_{\bar{\ell}-1}\} \subseteq B$ denote the image set of $f$ and $\bar{\ell} = |\im(f)|$. Define $A_i := \{ x \in A: f(x) = b_i \}$ and $\tau_i = | A_i |$ for $ 0 \leq i \leq \bar{\ell}-1$. Let $\calP$ be the set of all the preimage sets, i.e.,  $\calP=\{A_0, A_1, \ldots, A_{\bar{\ell}-1} \}$. Clearly, $\calP$ constitutes a partition of $A$. Furthermore, by the ZDB property, for
each $0 \leq i \leq \bar{\ell}-1$, the list of differences $a - a'$ with $a, a' \in A_i$ and $a \ne a'$, covers all nonzero elements of $A$ exactly $\lambda$ times. In this case, the set $\calP$ is called an $(n, \{\tau_0, \tau_1, \ldots, \tau_{\bar{\ell}-1}\}, \lambda)$-{\em partitioned difference family} (PDF). Because of the connection with PDF, each ZDB function can be identified with parameters $(n, \{\tau_0,
\tau_1, \ldots, \tau_{\bar{\ell}-1} \}, \lambda)$ (all these parameters are needed in some applications, see Section~\ref{sec-app}). We also associate every ZDB function with the three parameters $(n, \bar{\ell}, \lambda)$ since in some cases the parameters $\{\tau_0, \tau_1, \ldots, \tau_{\bar{\ell}-1}\}$ may not be available. 

Zero-difference balanced functions were first introduced by Ding in constructing optimal constant composition codes~\cite{Ding08} and optimal and perfect difference systems of sets~\cite{Ding09}. In the literature, perfect nonlinear functions~\cite{Hou,Feng,Ny91, ZGY,ZKW} and difference balanced functions~\cite{PW13,Zhou12} are special types of ZDB functions. ZDB functions unify different subjects in combinatorics, algebra and finite geometry, and they have found applications not
only in these three areas but also in communications, coding theory and cryptography. Due to their applications, ZDB functions have received a lot of attention recently. Besides constant composition codes and difference systems of sets, they can also be employed to construct optimal constant weight codes~\cite{WZ12,Zhou12} and optimal sets of frequency hopping sequences~\cite{GFM06,GMY09,WZ12}. In Table \ref{1:zdb} below, we summarize some known ZDB functions with parameters $(n,\bar{\ell},\lambda)$, and also the parameters $\{ \tau_0, \tau_1, \ldots,
\tau_{\bar{\ell}-1} \}$ if they are available. In tables below all the variables are positive integers, $q$ is always a prime power, and $p_i$'s are primes.

\begin{table}[ht]
 \caption{Some known ZDB functions with parameters $(n,\bar{\ell}, \lambda)$}\label{1:zdb}

\begin{center}{
\begin{tabular}{|c|c|c|c|}
  \hline
  \multicolumn{2}{|c|}{Parameters} & \multirow{2}{*}{Constraints} & \multirow{2}{*}{References} \\
  $(n, \bar{\ell}, \lambda)$ & $\{\tau_0, \tau_1, \ldots, \tau_{\bar{\ell}-1}\}$ &  & \\
  \hline
  \hline
  $(p^r,p^s, p^{r-s})$ &  & $p$ is a prime, $0 \le s \le r$ & \cite{Ny91}  \\
  \hline
  $(p^2,p,p)$  & $\{2p-1, p-1, \ldots, p-1\}$  & $p$ is an odd prime & \cite{Ding09}  \\
  \hline
  $\left(\frac{q^r-1}{N}, q, \frac{q^{r-1}-1}{N}\right)$ &  & $N|(q-1)$, and $\gcd(N,r)=1$ & \cite{Ding08,Ding09} \\
  \hline
  $(q^2+1, q, q+1)$ &    &  $q=2^s, s \ge 1$ & \cite{Ding08} \\
  \hline
  $(q-1, d, \frac{q-d}{d})$ & $\{ \frac{q}{d} - 1, \frac{q}{d}, \ldots, \frac{q}{d} \}$  &  $d|q$ & \cite{DT12} \\
  \hline
  $(q^r-1, q^s, q^{r-s}-1)$ &  &  $1 \le s \le r$ & \cite{Zhou12} \\
  \hline
  $\left(t\frac{q^r-1}{N}, q^s, t\frac{q^{r-s}-1}{N}\right)$ &  & $\begin{array}{c} N|(q-1), \mbox{ and } \gcd(N,r)=1 \\ 1 \le t \le N,\ 1 \le s \le r \end{array}$ & \cite{Zhou12} \\
  \hline
  $(n, \frac{n + e -1}{e}, e-1)$ & $\{ 1,e,\ldots, e \}$  &  $\begin{array}{c} n=p_1^{m_1}p_2^{m_2} \cdots p_k^{m_k}, 2 <p_1<p_2<\cdots<p_k\\
  \mbox{ and } e|(p_i-1) \mbox{ for } 1 \le i \le k
  \end{array}$ & \cite{Cai13} \\
  \hline
  $(n, \frac{n + e -1}{e}, e-1)$ & $\{1,e, \ldots, e\}$  &  $\begin{array}{c} n=p_1^{m_1}p_2^{m_2} \cdots p_k^{m_k}, p_1<p_2<\cdots<p_k\\
  \mbox{ and } e|(p_i^{m_i}-1) \mbox{ for } 1 \le i \le k
\end{array}$ & Theorem~\ref{1:main} \\
  \hline
  $(2^m-1, \frac{2^m + m - 2}{m}, m-1)$ &  $\{ 1, m, \ldots, m \}$    &   $m$ is a prime & Theorem~\ref{2:zdb} \\
  \hline
  $(2^m-1, \frac{2^{m-1} + m - 1}{m}, 2m-1)$ & $\{1, 2m, \ldots, 2m \}$  &   $m$ is an odd prime & Theorem~\ref{3:zdb} \\
  \hline
\end{tabular}}
\end{center}
\end{table}

Very recently, in~\cite{Cai13},  Cai, Zeng, Helleseth, Tang, and Yang constructed $(n, (n + e -1)/e, e-1)$-ZDB functions on $(\Z_n, +)$, where $n$ is odd and has the canonical factorization
\[n=p_1^{m_1} p_2^{m_2} \cdots p_k^{m_k}, \quad 2<p_1<p_2<\cdots<p_k,\]
and $e>1$ such that $e|(p_i-1)$ for all $1 \le i \le k$. Their construction employs the tool of generalized cyclotomy in the rings $\Z_n$, and generates many ZDB functions with new parameters.

In this paper, inspired by the idea of \cite{Cai13} and utilizing generalized cyclotomy in the rings $\gf(p_1^{m_1}) \times \cdots \times \gf(p_k^{m_k})$, we construct $(n, (n+e-1)/e, e-1)$-ZDB functions on the abelian groups $(\gf(p_1^{m_1}) \times \cdots \times \gf(p_k^{m_k}), +)$, with
\[n=p_1^{m_1} p_2^{m_2} \cdots p_k^{m_k}, \quad p_1<p_2<\cdots<p_k,\]
and $e>1$ such that $e|(p_i^{m_i}-1)$ for all $1 \le i \le k$.

One aspect of difference between \cite{Cai13} and our construction is that the groups $\Z_n$ are cyclic, while in our case the groups $(\gf(p_1^{m_1}) \times \cdots \times \gf(p_k^{m_k}), +)$ are not cyclic in general. This is an advantage of \cite{Cai13} over our construction, because ZDB functions on cyclic groups have more applications than those on noncyclic groups. On the other hand, our construction provides many new and more flexible parameters compared with \cite{Cai13}, since in our
construction $n$ can be even and the requirement $e|(p_i^{m_i}-1)$ gives more flexibility. For example, take $n=7^2 \cdot 13^2$, then $\gcd(7^2-1,13^2-1)=24$. So for $e>1$ and $e|24$, we can take $e=2,3,6,4,8,12,24$. In comparison with the construction in~\cite{Cai13}, however, since $\gcd(7-1,13-1)=6$, the requirement $e>1$ and $e|6$ only allows us to take $e=2,3,6$. The two constructions may overlap only when $m_1=\ldots=m_k=1$, i.e., when $n$ is a square-free positive integer.

In addition to this construction, we propose two other families of ZDB functions on the cyclic groups $\Z_n$ with new parameters based on $2$-cyclotomic cosets, where $n=2^m-1$ for any prime $m$. It may be noted that these constructions cannot be generalized to $p$-cyclotomic cosets on $\Z_n$ where $n=p^m-1$ for $p>2$.

The rest of the paper is organized as follows. In Section~\ref{sec-first}, we present the first construction of ZDB functions
on the groups $(\gf(p_1^{m_1}) \times \cdots \times \gf(p_k^{m_k}), +)$. In Section~\ref{sec-2cyclo}, we describe the second and third
construction of ZDB functions on the groups $(\Z_n, +)$.  In Section~\ref{sec-app}, we summarize some applications of the
ZDB functions. In Section~\ref{sec-con}, we make some remarks.

\section{The first family of $(n, (n+e-1)/e, e-1)$-ZDB functions on the Abelian groups $(\gf(p_1^{m_1}) \times \cdots \times \gf(p_k^{m_k}), +)$}\label{sec-first}

The first construction of ZDB functions is described as follows. Let $q_1,\ldots,q_k$ be distinct prime powers and let
\[n=q_1 q_2 \cdots q_k.\]
For each $i$, let $\gf(q_i)$ be a finite field of order $q_i$ and $g_i$ be a generator of the multiplicative group $\gf(q_i)^*:=\gf(q_i) \setminus \{0\}$. Consider the ring
\[A=\gf(q_1) \times \gf(q_2) \times \cdots \times \gf(q_k). \]
For each non-empty subset $I \subseteq \{1,\ldots,k\}$, let
\[A_{I}=\left\{\vec{x}=(x_1,\ldots,x_k) \in A: \begin{array}{lc}
x_i \in \gf(q_i)^*,& \mbox{ if } i \in I,\\
x_i =0,& \mbox{ if } i \notin I.
\end{array}\right\}\]
Without confusion we may identify $A_I$ with $\prod_{i \in I}\gf(q_i)^*$. Then $A_I$ is a multiplicative group with identity $\vec{1}=(1,\ldots,1)$.

For each $e>1$ with $e|(q_i-1)$ for all $1 \le i \le k$, let
\[q_i-1=e \cdot f_i, \quad 1 \le i \le k\]
and let $\vec{g_I} \in A_I$ be given by
\[\vec{g_I}=\left(g_i^{f_i}\right)_{i \in I}.\]
Since the order of $g_i$ is $q_i-1$ for each $i$, the order of $\vec{g_I} \in A_I$ is $e$. Let $D_I \subseteq A_I$ be the cyclic subgroup generated by $\vec{g_I}$, then $|D_I|=e$. Clearly,
\[\left|\left(D_{I} - \vec{1} \right) \bigcap A_{I}\right|=e-1.\]
We can decompose $A_I$ into a disjoint union of left cosets of $D_I$ as
\begin{eqnarray} \label{1:ai}
A_I=\coprod_{\alpha_I \in R_I} \alpha_I D_I,
\end{eqnarray}
where $R_I \subseteq A_I$ is a fixed set of representatives for $A_I/D_I$. We find that \[|R_I|=|A_I|/|D_I|=\frac{1}{e} \cdot \prod_{i \in I}(q_i-1). \]

Let
\[\mathcal{S}:=\biggl\{\alpha_I D_I: \forall \,\, \alpha_I \in R_I, \forall \,\, \emptyset \ne I \subseteq \{1,\ldots,k\}  \biggr\} \bigcup \bigl\{\{\vec{0}\}\bigr\}. \]
The set $\mathcal{S}$ has order $\frac{n-1}{e}+1$. Let $\eta(\cdot)$ be any bijection from $\mathcal{S}$ to $\Z_{\frac{n-1}{e}+1}$. We define $f: A \to \Z_{\frac{n-1}{e}+1}$ by
\[f(\vec{x})=\left\{\begin{array}{lll}
\eta\left(\alpha_I D_I\right) &:& \mbox{ if } \vec{x} \in \alpha_I D_I \mbox{ for some } \alpha_I \in R_I, \emptyset \ne I \subseteq \{1,\ldots,k\},\\
\eta\left(\{\vec{0}\}\right) &:& \mbox{ if } \vec{x}=\vec{0}=(0,\ldots,0).
\end{array}\right.\]
This is well-defined, because each non-zero vector $\vec{x} \in A$ belongs to some $A_I$ for a unique non-empty $I$, and by the decomposition (\ref{1:ai}), $\vec{x}$ belongs to some $\alpha_I D_I$ for a unique $\alpha_I \in R_I$. Since $|\alpha_I D_i| = e$ for each $\alpha_I \in R_I$, it is easily seen that the sizes of the preimage sets of $f$ are $\{1, e, \ldots, e\}$.

\begin{thm}\label{1:main}
  The function $f$ defined above is an $(n, \frac{n + e - 1}{e}, e-1)$-ZDB function from $A$ onto $\Z_{\frac{n-1}{e}+1}$.
\end{thm}

\begin{proof} For each $\vec{a}=(a_i)_i \in A \setminus \{\vec{0}\}$, we may assume that $\vec{a} \in A_I$ for a unique non-empty set $I \subseteq \{1,\ldots,k\}$. By definition, $f(\vec{x}+\vec{a})=f(\vec{x})$ if and only if $\vec{x}$ and $\vec{x}+\vec{a}$ belong to the same set in $\mathcal{S}$. This implies that $\vec{x} \ne \vec{0}$. Say $\vec{x}=(x_i)_i$ and $\vec{x}+\vec{a}=(x_i+a_i)_i$ belong to the same set $\alpha D_{I'}$ for some $\alpha=(\alpha_i)_i \in
  R_{I'}$ and some non-empty set $I'$. This means that there exist $0 \le t,s \le e-1$ such that
\[\begin{array}{ll}
x_i=\alpha_i g_i^{f_it}, \quad x_i+a_i=\alpha_i g_i^{f_is} & \forall \, i \in I', \\
x_i=0, \quad x_i+a_i=0 & \forall \, i \notin I'.
\end{array}\]
In particular, we have $a_i=0$ for all $i \notin I'$, hence $I \subseteq I'$. If $I' \ne I$, then we can find $i_1 \in I' \setminus I$, and $a_{i_1}=0$ since $\vec{a} \in A_I$, and
\[x_{i_1}=\alpha_{i_1}g_{i_1}^{f_{i_1}t}= x_{i_1}+a_{i_1}=\alpha_{i_1}g^{f_{i_1}s}_{i_1}.  \]
This implies that
\[\alpha_{i_1}g_{i_1}^{f_{i_1}t} =\alpha_{i_1}g^{f_{i_1}s}_{i_1} \Longrightarrow t=s. \]
Thus we find $a_i=0$ for all $i \in I'$. We already know that $a_i=0$ for all $i \notin I'$. This means $\vec{a}=\vec{0}$, a contradiction. Therefore we must have $I'=I$ where $\vec{a} \in A_{I}$. So
\begin{eqnarray*}
  \lefteqn{|\{\vec{x} \in A: f(\vec{x}+\vec{a}) f(\vec{x})\}| } \\
  & = & \sum_{\emptyset \ne I' \subseteq \{1,\ldots,k\}} |\{\vec{x} \in A_{I'}: f(\vec{x}+\vec{a})=f(\vec{x})\}| \\
& =  & |\{\vec{x} \in A_{I}: f(\vec{x}+\vec{a})=f(\vec{x})\}| \\
& = & \sum_{\alpha \in R_{I}}\left|\alpha D_{I} \bigcap \left(\alpha D_{I}-\vec{a}\right)\right|.
\end{eqnarray*}

Every element in the set $\alpha D_{I} \bigcap \left(\alpha D_{I}-\vec{a}\right)$ corresponds one-to-one to unique $\vec{x},\vec{y} \in D_{I}$ such that $\alpha \vec{x}-\vec{a}=\alpha \vec{y}$, or equivalently $\vec{1}-\vec{a} \alpha^{-1} \vec{x}^{-1}= \vec{x}^{-1} \vec{y}$ with $\vec{x}^{-1}, \vec{x}^{-1}\vec{y} \in D_{I}$. Here for $\vec{x} \in A_{I}$, $\vec{x}^{-1}$ denotes the multiplicative inverse of $\vec{x}$ in $A_{I}$. So we have
\begin{eqnarray*}
\left|\alpha D_{I} \bigcap \left(\alpha D_{I}-\vec{a}\right)\right|=\left|D_{I} \bigcap \left( \vec{1}-\vec{a} \alpha^{-1}D_{I}\right)\right|=\left|\left(D_{I}-\vec{1}\right) \bigcap \left(-\vec{a} \alpha^{-1}\right) D_{I}\right|.
\end{eqnarray*}
It is easy to observe that as $\alpha$ runs over $R_{I}$, a set of representatives for $A_{I}/D_{I}$, the element $-\vec{a} \alpha^{-1}$ will also run over a set of representatives for $A_{I}/D_{I}$. Therefore we obtain
\begin{eqnarray*}
|\{\vec{x} \in A: f(\vec{x}+\vec{a})=f(\vec{x})\}|&=&\sum_{\alpha \in R_{I}}\left|\left(D_{I}-\vec{1}\right) \bigcap \alpha D_{I}\right|=\left|\left(D_{I}-\vec{1}\right) \bigcap A_{I}\right|=e-1.
\end{eqnarray*}
This completes the proof of Theorem \ref{1:main}.
\end{proof}

\section{Two more families of ZDB functions on $(\Z_n,+)$ from $2$-cyclotomic cosets modulo $n$}\label{sec-2cyclo}

In this section, employing $2$-cyclotomic cosets modulo $n=2^m-1$, we present two families of ZDB functions on $(\Z_n,+)$  with new parameters. The ZDB functions in one family have parameters $(2^m-1, (2^m + m - 2)/m, m-1)$, and those in the
other family have parameters $(2^m - 1, (2^{m-1} + m - 1)/m, 2m - 1)$. Furthermore, the parameters $\{ \tau_0, \tau_1, \ldots, \tau_{\bar{\ell}-1} \}$, i.e., the sizes of the preimage sets,  of the ZDB functions are also determined.

Let $n=2^m-1$. The $2$-cyclotomic coset modulo $n$ containing $i$ is defined by
$$
\{i, i \times 2 \bmod{n}, i \times 2^2 \bmod{n}, \cdots, i \times 2^{\ell_i} \bmod{n}\} \subset \Z_n,
$$
where $\ell_i$ is the least positive integer such that $i \equiv i 2^{\ell_i} \pmod{n}$, and is called the size of this
$2$-cyclotomic coset. The leader of a $2$-cyclotomic coset modulo $n$ is the least integer in the $2$-cyclotomic
coset. Clearly, all the $2$-cyclotomic cosets modulo $n$ form a partition of $\Z_n$. It is noted that $n=2^m-1$ may not be a prime when $m$ is a prime. For example, $n=2^{11}-1=23 \times 89$.

\subsection{The family of $(2^m-1, (2^m + m - 2)/m, m-1)$-ZDB functions on $(\Z_n, +)$}\label{sec-second}

Let $m$ be a prime, and let $n=2^m-1$. Since $m$ is a prime, every nonzero $2$-cyclotomic coset has size $m$, and the total number of nonzero $2$-cyclotomic cosets modulo $n$ is equal to $(2^m-2)/m$.
Let $\Gamma_m$ denote the set of all $2$-cyclotomic coset leaders. Then
$$
|\Gamma_m|=1+ \frac{2^m-2}{m} = \frac{2^m + m - 2}{m}.
$$

We now define a function $f$ from $(\Z_n, +)$ to itself by
$$
f(x)=i_x,
$$
where $i_x$ is the coset leader of the $2$-cyclotomic coset containing $x$.
Since every nonzero $2$-cyclotomic coset has $m$ elements modulo $2^m - 1$, by definition, the sizes of the preimage sets of $f$ form the set $\{1, m, \ldots, m\}$.

\begin{theorem} \label{2:zdb}
  Let $m$ be a prime. Then the function $f$ defined above is a $(2^m-1, \frac{2^m + m - 2}{m}, m-1)$-ZDB function on $(\Z_n, +)$.
\end{theorem}

\begin{proof} Note that $|\im(f)|=|\Gamma_m| = \frac{2^m + m - 2}{m}$. It suffices to prove that for every $a \not \equiv 0 \pmod{2^m-1}$, the number of $x$ with $1 \le x \le 2^m-1$ such that $x+a$ and $x$ belong to the same $2$-cyclotomic set is always $m-1$. The existence of such an $x$ means that there is an integer $k$ with $1 \le k \le m - 1$, such that
\[x+a \equiv 2^kx \pmod{2^m-1}, \]
or equivalently,
\begin{eqnarray*} (2^k-1)x \equiv a \pmod{2^m-1}. \end{eqnarray*}
Since $m$ is a prime and $k<m$, we have
\[\gcd(2^k-1,2^m-1)=2^{\gcd(k,m)}-1 = 1.  \]
We denote by $\overline{2^k-1}$ the multiplicative inverse of $2^k-1$ modulo $2^m-1$. Thus,
\[x \equiv (\overline{2^k-1}) \cdot a \pmod{2^m-1}, \]
and this holds for all $1 \le k \le m-1$. It is also clear that $2^k-1 \not \equiv 2^l-1 \pmod{2^m-1}$ for $1 \le k \ne l \le m-1$, hence the number of such $x$ is always $m-1$. This completes the proof of Theorem \ref{2:zdb}.
\end{proof}

\subsection{The family of $(2^m-1, (2^{m-1} + m - 1)/m, 2m-1)$-ZDB functions on $(\Z_n, +)$}\label{sec-third}

Let $m$ be an odd prime and let $n=2^m-1$. Same as in Section~\ref{sec-second}, let $\Gamma_m$ denote the set of all $2$-cyclotomic coset leaders and further $\Pi_m$ denote the set of all $2$-cyclotomic cosets modulo $n$. Since $m$ is prime, every nonzero $2$-cyclotomic cosets modulo $n$ has the size $m$ and $|\Gamma_m|=1+ (2^m - 2)/m $.
Define
$$
\Delta_m=\{B \cup (-B): B \in \Pi_m  \} ,
$$
where $-B=\{n-i: i \in B\}$. Similarly, the leader of any $B \cup (-B)$ is the least integer in this set. It is easy to prove that $B$ and $-B$ are disjoint for each $ \{0 \} \ne B \in \Pi_m$, and hence
$$
|\Delta_m|= 1 + \frac{2^{m-1}-1}{m} = \frac{2^{m-1} + m - 1}{m}.
$$

We now define a function $g$ from $(\Z_n, +)$ to itself by
$$
g(x)=j_x,
$$
where $j_x$ is the leader of the set $B \cup (-B)$  containing $x$.
Since every nonzero set $B \cup (-B)$ has $2m$ elements, the sizes of the preimage sets of $g$ form the set $\{1, 2m, \ldots, 2m\}$.

\begin{theorem} \label{3:zdb}
  Let $m$ be an odd prime. Then the function $g$ defined above is a $(2^m-1, \frac{2^{m-1} + m - 1}{m}, 2m-1)$-ZDB function on $(\Z_n, +)$.
\end{theorem}

\begin{proof} Note that $|\im(g)|=|\Delta_m| = \frac{2^{m-1} + m - 1}{m}$. We only need to prove that for each $a \not \equiv 0 \pmod{2^m-1}$, the number of $x$ with $1 \le x \le 2^m-1$ such that $x+a$ belongs to the $2$-cyclotomic set that contains either $x$ or $-x$ is always $2m-1$. The existence of such an $x$ means that there is an integer $k$ with $1 \le k \le m - 1$, such that
\begin{eqnarray} \label{3:con1} x+a \equiv 2^kx \pmod{2^m-1}, \end{eqnarray}
or there is an integer $t$ with $1 \le t \le m $ such that
\begin{eqnarray} \label{3:con2} x+a \equiv -2^tx \pmod{2^m-1}. \end{eqnarray}
As for (\ref{3:con1}), similar to the proof of Theorem \ref{2:zdb}, the number of solutions for $x$ is $m-1$. As for (\ref{3:con2}), we get
\[(2^t+1)x \equiv -a \pmod{2^m-1}. \]
Notice that if $t<m$, since $m$ is an odd prime, we have
\[\gcd(2^{2t}-1,2^m-1)=2^{\gcd(2t,m)}-1=1, \]
and further $\gcd(2^t+1,2^m-1) = 1$.
If $t=m$, we have
\[\gcd(2^m+1,2^m-1)=\gcd(2,2^m-1)=1. \]
Hence, $2^t+1 $ is invertible modulo $2^m - 1$ for all $1 \le k \le m$. It then follows that the number of $x$ satisfying (\ref{3:con2}) is $m$. On the other hand, (\ref{3:con1}) and (\ref{3:con2}) can not be satisfied simultaneously, because otherwise we obtain for some $1 \le k \le m-1$, $1 \le t \le m$
\[(2^t+2^k) a \equiv 0 \pmod{2^m-1}.\]
However, we note that
\[\gcd(2^t+2^k,2^m-1)=1, \]
implying that $a \equiv 0 \pmod{2^m-1}$, which is a contradiction to the assumption $a \not\equiv 0 \pmod{2^m - 1}$. Then the total number of $x$ satisfying either (\ref{3:con1}) or (\ref{3:con2}) is $2m-1$. This completes the proof.
\end{proof}

\section{Two applications of the ZDB functions presented in this paper}\label{sec-app}

In this section, we deal with the applications of the ZDB functions of this paper in constant composition codes and difference systems of sets.

\subsection{Optimal constant composition codes}\label{sec-app-ccc}

Let $\F_\ell$ denote the set $\{0, 1, \ldots, \ell - 1\}$ (also called {\em alphabet}), and let $\F_\ell^n$ be the set of all $n$-tuples over $\F_\ell$ (also called {\em words}). An $(n, M, d, w)_\ell$ {\em constant weight code} (CWC) is a code $C \subset \F_\ell^n$ with size $M$ and minimum Hamming distance $d$ such that the Hamming weight of each codeword is $w$. An $(n, M, d, [w_0, w_1, \ldots, w_{\ell - 1}])_\ell$ {\em constant composition code} (CCC) is a code $C
\subset \F_\ell^n$ with size $M$ and minimum Hamming distance $d$ such that in every codeword the element $i$ appears exactly $w_i$ times for every $i \in \F_\ell$. An $(n, M, d, [w_0, w_1, \ldots, w_{\ell - 1}])_\ell$ CCC is called a {\em permutation code} if $n = \ell$ and $w_i = 1$ for all $i \in \F_\ell$. By definition, constant composition codes are a special class of constant weight codes and permutation codes are a further special class of constant composition codes.

Let $A_\ell (n, d, [w_0, w_1, \ldots, w_{\ell - 1}])$ denote the maximum size of an $(n, M, d, [w_0, w_1, \ldots, w_{\ell - 1}])_\ell$ CCC. The following upper bound on the maximum size of a CCC was derived in~\cite{Luo03}.

\begin{lemma}\label{lem-cccbound}
  If $nd - n^2 + (w_0^2 + w_1^2 + \cdots + w_{\ell - 1}^2) > 0 $,
  \begin{equation}\label{eqn-cccbound}
 A_\ell (n, d, [w_0, w_1, \ldots, w_{\ell - 1}]) \leq \frac{nd}{nd - n^2 + (w_0^2 + w_1^2 + \cdots w_{\ell - 1}^2)}.
  \end{equation}
\end{lemma}

An $ (n, M, d, [w_0, w_1, \ldots, w_{\ell - 1}])_\ell$ constant composition code is said to be {\em optimal} if the bound of (\ref{eqn-cccbound}) is met. In~\cite{DY05,Ding08}, the link between ZDB functions and optimal CCCs was established, and PDFs and ZDB functions were used to construct optimal CCCs.

\begin{lemma}\label{lem-zdbccc}
  Suppose that $f$ is an $(n, \bar{\ell}, \lambda)$-ZDB function from an abelian group $(A, +)$ of order $n$ to an abelian group $(B, +)$ of order $\ell$ and $\im(f)$ is the image set of $f$ with $|\im(f)| = \bar{\ell}$. Let $A = \{a_0, a_1, \ldots, a_{n-1}\}$ and $\im(f) = \{b_0, b_1, \ldots, b_{\bar{\ell} - 1}\}$. Define $\tau_i = |\{ x \in A: f(x) = b_i \}|$ for $0 \leq i \leq \bar{\ell}-1$. Then the code
 $$
 \mathcal{C} = \{ ( f(a_0 + a_i), \ldots, f(a_{n-1} + a_i)): 0 \leq i \leq n-1 \}
 $$
 is an $(n, n, n - \lambda, [\tau_0, \tau_1, \ldots, \tau_{\bar{\ell} - 1}])_{\bar{\ell}}$ CCC over $\im(f)$ meeting the bound of (\ref{eqn-cccbound}).
\end{lemma}

\begin{table}[ht]
  \caption{Some known optimal CCCs with parameters $(n,M,d, [w_0, w_1, \ldots, w_{\ell -1}])_\ell$
}\label{1:ccc}

\begin{center}{
\begin{tabular}{|c|c|c|}
  \hline
  Parameters & Constraints & References \\
  \hline
  \hline
  $\left(\frac{3^r - 1}{2}, \frac{3^r - 1}{2}, s, \left[ \frac{s-1}{2}, \frac{s - s^{1/2}}{2}, \frac{s + s^{1/2}}{2} \right] \right)_3$ & $r$ is odd, $s = 3^{r-1}$ & \cite{Yuan05} \\
  \hline
  $(q, q, q-r, [2r-1, 2, \ldots, 2])_r$ & $q \equiv 1 \pmod{4}$, $r = \frac{q+3}{4}$ & \cite{DY05} \\
  \hline
  $(q, q, q - s + 1, [s, \ldots, s, 1])_r$ & $q \equiv 1 \pmod{s}$, $r = \frac{q+s-1}{s}$ & \cite{DY05} \\
  \hline
   & $q \equiv 1 \pmod{2s}$ &  \\
   $(q, q, q - \frac{s-1}{2}, [s, \ldots, s, 1, \ldots, 1])_r$ &  $ r = \frac{q-1}{2s} + \frac{q+1}{2}$, &  \cite{DY05} \\
    &  $s$ appears $\frac{q-1}{2s}$ times &  \\
    \hline
    $(q(q+1), q^2, q^2, [q+1, \ldots, q+1])_q$ &    &  \cite{DY05} \\
    \hline
   $(q^{2r}, q^{2r}, (q-1)q^{2r-1}, [q^{2r-1} + (q-1)q^{r-1}, $  &   & \multirow{2}{*}{\cite{DY051}} \\
   $q^{2r-1} - q^{r-1}, \ldots, q^{2r-1} - q^{r-1}])_q$ &    &  \\
   \hline
   $\left( \frac{q^r - 1}{2}, \frac{q^r - 1}{2}, \frac{q^r - q^{r-1}}{2}, [\tau_0, \tau_1, \ldots, \tau_{q-1}]\right)_q $ &  $q$ is odd  &  \cite{DY051} \\
   \hline
   $(9s, 9^r, 6 \cdot 9^{r-1}, [5s, 2s, 2s])_3$ & $s = \frac{9^r - 1}{8}$ & \cite{Luo03} \\
   \hline
   $(8s, 8^r, 6 \cdot 8^{r-1}, [3s, 3s, 2s]_3$ & $s = \frac{8^r - 1}{7}$ & \cite{Luo03} \\
   \hline
   $(10s, 5^r, 7 \cdot 5^{r-1}, [6s, 2s, 2s])_3$ & $s = \frac{5^r - 1}{4}$ & \cite{Luo03} \\
   \hline
   $(qt, q^r, q^r, [t, \ldots, t])_q$ &  $t = \frac{q^r-1}{q-1}$ & \cite{Luo03} \\
   \hline
   $\left(qt, q^r, \frac{(q+3)q^{r-1}}{2}, \left[ \frac{(q-1)t}{2}, \frac{(q-1)t}{2}, t \right] \right)_3$ &  $t = \frac{q^r - 1}{q-1}$, $q$ is odd & \cite{Luo03} \\
   \hline
   $\left( \frac{q^r - 1}{N}, \frac{q^r - 1}{N}, \frac{q^r - q^{r-1}}{N}, [\tau_0, \tau_1, \ldots, \tau_{q-1}]\right)_q $ & $N| (q-1)$, $\gcd(N,r) = 1$ & \cite{Ding08} \\
   \hline
   $(q^2 + 1, q^2 + 1, q^2 - q, [\tau_0, \tau_1, \ldots, \tau_{q-1}])_q$ & $q = 2^s$, $s \ge 1$ & \cite{Ding08} \\
   \hline
   $(q-1, q-1, q - \frac{q-d}{d} - 1, [\frac{q}{d} -1, \frac{q}{d}, \ldots, \frac{q}{d}])_d$ & $d|q$ & \cite{DT12} \\
   \hline
   $(q^r - 1, q^r - 1, q^r - q^{r-s}, [\tau_0, \tau_1, \ldots, \tau_{q^s - 1}])_{q^s}$ & $1 \leq s \leq r$ & \cite{Zhou12} \\
   \hline
   $(t\frac{q^r-1}{N}, t\frac{q^r-1}{N}, t\frac{q^r - q^{r-s}}{N}, [\tau_0, \tau_1, \ldots, \tau_{q^s-1}])_{q^s}$ & $\begin{array}{c} N|(q-1),\  \gcd(N,r) = 1 \\
     1 \leq t \leq N,\ 1 \leq s \leq r \end{array}$ & \cite{Zhou12} \\
   \hline
   $(n, n, n - e + 1, [1, e,  \ldots, e])_{\frac{n + e - 1}{e}}$    &  $\begin{array}{c} n=p_1^{m_1}p_2^{m_2} \cdots p_k^{m_k}, 2< p_1<p_2<\cdots<p_k\\
  \mbox{ and } e|(p_i-1) \mbox{ for } 1 \le i \le k
\end{array}$ & \cite{Cai13} \\
   \hline
   $(n, n, n - e + 1, [1, e,  \ldots, e])_{\frac{n + e - 1}{e}}$    &  $\begin{array}{c} n = p_1^{m_1}p_2^{m_2} \cdots p_k^{m_k}, p_1<p_2<\cdots<p_k\\
  \mbox{ and } e|(p_i^{m_i}-1) \mbox{ for } 1 \le i \le k
\end{array}$ & Theorem~\ref{1:main} \\
  \hline
  $(2^m-1, 2^m - 1, 2^m - m, [ 1, m, \ldots, m ])_{\frac{2^m + m - 2}{m}}$    &   $m$ is a prime & Theorem~\ref{2:zdb} \\
  \hline
  $(2^m-1, 2^m -1, 2^m - 2m, [1, 2m, \ldots, 2m])_{\frac{2^{m-1}+m-1}{m}} $  &   $m$ is an odd prime & Theorem~\ref{3:zdb} \\
  \hline
\end{tabular}}
\end{center}
\end{table}

We remark that every ZDB function corresponds to an optimal CCC using this standard method in Lemma~\ref{lem-zdbccc}. In Table~\ref{1:ccc} we summarize some known optimal CCCs with parameters $(n, M, d, [w_0, w_1, \ldots, w_{\ell - 1}])_\ell$, including the new parameters of the CCCs obtained from the three new families of ZDB functions of this paper.

\subsection{Optimal and perfect difference systems of sets}\label{sec-app-dss}

Difference systems of sets (DSS) were introduced by Levenstein~\cite{Lev71} (see also~\cite{Lev04}) for the construction of comma-free codes for synchronization. Let $n$ be a positive integer, and let $\Z_n$ be the integer ring modulo $n$. An $(n, \{\tau_0, \tau_1, \ldots, \tau_{\ell-1}\}, \rho)$ {\em difference system of set} (DSS) is a collection of $\ell$ disjoint sets $D_i \subseteq \Z_n$ such that $|D_i| = \tau_i$ for all $0 \leq i < \ell $ and the multiset
\begin{equation}\label{eqn-dssmulti}
\{*  ( b - b') \bmod{n} : b \in D_i,\ b' \in D_j, \ i \ne j,, \ 0 \leq i, j \leq \ell -1 *\}
\end{equation}
contains every nonzero element $x \in \Z_n$ at least $\rho$ times. A DSS is called {\em perfect} if every nonzero element $x \in \Z_n$ is contained exactly $\rho$ times in the multiset of (\ref{eqn-dssmulti}). A DSS is said {\em regular} if all the subsets $D_i$'s are of the same size.

For the application of DSS to code synchronization, the number
$$
r_\ell (n, \rho) = \sum_{i=0}^{\ell - 1} | D_i |
$$
is required to be as small as possible. A lower bound on $r_\ell (n, \rho)$ is the following~\cite{Wang06}.

\begin{lemma}\label{lem-dss}
  For any DSS with parameters $(n, \{\tau_0, \tau_1, \ldots, \tau_{\ell-1}\}, \rho)$,
  \begin{equation}\label{eqn-dssbound}
    r_\ell (n, \rho) \geq \sqrt{\sq\left( \rho(n-1) + \left\lceil \frac{\rho (n-1)}{\ell - 1} \right\rceil \right)} ,
  \end{equation}
where $\sq(x)$ denotes the smallest square number that is no less than the positive integer $x$, and $\lceil x \rceil$ denotes the ceiling function.
\end{lemma}
A perfect $(n, \{\tau_0, \tau_1, \ldots, \tau_{\ell-1}\}, \rho)$ DSS is called {\em optimal} if the bound of (\ref{eqn-dssbound}) is met. The correspondence between ZDB functions and perfect DSSs was first established in~\cite{Ding09} (see also~\cite{Zhou12}).

\begin{lemma}\label{lem-zdbdss}
  Suppose that $f$ is an $(n, \bar{\ell}, \lambda)$-ZDB function from $(\Z_n, +)$ to an abelian group $(B, +)$ of order $\ell$ and $\im(f)$ is the image set of $f$ with $|\im(f)| = \bar{\ell}$. Let $\im(f) = \{b_0, b_1, \ldots, b_{\bar{\ell} - 1}\}$. Define $D_i = \{ x \in \Z_n : f(x) = b_i\}$, and $\tau_i = |D_i|$ for $0 \leq i \leq \bar{\ell} - 1$. Then the set
  $$
  \mathcal{D} =  \{ D_i : 0 \leq i \leq \bar{\ell} - 1\}
  $$
  is an $(n, \{\tau_0, \tau_1, \ldots, \tau_{\bar{\ell} - 1}\}, n - \lambda)$ perfect DSS. Furthermore, if $\bar{\ell} \lambda \leq n$, $\mathcal{D}$ is optimal with respect to the bound of (\ref{eqn-dssbound}).
\end{lemma}

\begin{table}[ht]
  \caption{Some known optimal and perfect DSSs with parameters $(n, \{\tau_0, \tau_1, \ldots, \tau_{\ell -1}\}, \rho)$
}\label{1:dss}

\begin{center}{
\begin{tabular}{|c|c|c|}
  \hline
  Parameters & Constraints & References \\
  \hline
  \hline
  $\left( \frac{q^r - 1}{N}, \{\tau_0, \tau_1, \ldots, \tau_{q-1}\}, \frac{q^r - q^{r-1}}{N} \right) $ & $N| (q-1)$, $\gcd(N,r) = 1$ & \cite{Ding08} \\
   \hline
   $(q^2 + 1, \{ \tau_0, \tau_1, \ldots, \tau_{q-1} \}, q^2 - q )$ & $q = 2^s$, $s \ge 1$ & \cite{Ding08} \\
   \hline
   $(p^2, \{2p-1, p-1, \ldots, p-1\}, p^2 - p)$ & $p$ is a prime & \cite{Ding09} \\
   \hline
   $(q-1,   \{\frac{q}{d} -1, \frac{q}{d}, \ldots, \frac{q}{d}\}, q - \frac{q-d}{d} - 1)$ & $d|q$ & \cite{DT12} \\
   \hline
   $(q^r - 1, \{\tau_0, \tau_1, \ldots, \tau_{q^s - 1}\}, q^r - q^{r-s})$ & $1 \leq s \leq r$ & \cite{Zhou12} \\
   \hline
   $(t\frac{q^r-1}{N}, \{\tau_0, \tau_1, \ldots, \tau_{q^s-1}\}, t\frac{q^r - q^{r-s}}{N})$ & $\begin{array}{c} N|(q-1),\  \gcd(N,r) = 1 \\ 1 \leq t \leq N,\ 1 \leq s \leq r\end{array}$ & \cite{Zhou12} \\
   \hline
   $(n, \{1, e,  \ldots, e\}, n - e + 1 )$    &  $\begin{array}{c} n=p_1^{m_1}p_2^{m_2} \cdots p_k^{m_k}, 2< p_1<p_2<\cdots<p_k\\
  \mbox{ and } e|(p_i-1) \mbox{ for } 1 \le i \le k, n \ge (e-1)^2
\end{array}$ & \cite{Cai13} \\
  \hline
  $(2^m-1, \{1, m, \ldots, m \}, 2^m - m )$    &   $m$ is a prime & Theorem~\ref{2:zdb} \\
  \hline
  $(2^m-1, \{1, 2m, \ldots, 2m\}, 2^m - 2m )$  &   $m$ is an odd prime, $m \ge 11$ & Theorem~\ref{3:zdb} \\
  \hline
\end{tabular}}
\end{center}
\end{table}

We emphasize that the DSSs constructed from ZDB functions using Lemma~\ref{lem-zdbdss} may not be optimal unless the condition $\bar{\ell} \lambda \leq n$ is satisfied. It is easy to check that the DSSs given by the ZDB functions in Theorem~\ref{2:zdb} are optimal for all primes $m$, and those constructed by the ZDB functions in Theorem~\ref{3:zdb} achieve optimality for primes $m \ge 11$.

Note that the group $(\gf(p_1^{m_1}) \times \cdots \times \gf(p_k^{m_k}), +)$ is cyclic only when $m_1=\cdots=m_k=1$, that is, $n=p_1 \cdots p_k$ is square-free. The ZDB functions in Theorem~\ref{1:main} can be employed to construct DSSs only when $m_1=\cdots=m_k=1$.

In Table~\ref{1:dss}, we summarize the parameters of some known optimal and perfect DSSs, including the parameters of the DSSs obtained from the two families of ZDB functions in Section~\ref{sec-2cyclo}.

\section{Concluding remarks}\label{sec-con}

In this paper, we present three new families of ZDB functions with parameters $\{n,\{\tau_0,\ldots,\tau_{\bar{\ell}-1}\},\bar{\ell},\lambda\}$. The parameters of optimal constant composition codes, optimal and perfect difference systems of sets obtained from these new families of ZDB functions are also summarized.

As we have seen, with respect to applications in constant composition codes and difference systems of sets, every parameter in the set
$\{n,\{\tau_0,\ldots,\tau_{\bar{\ell}-1}\},\bar{\ell},\lambda\}$ makes a difference. Hence, when comparing the parameters
of two ZDB functions, it may be more appropriate to compare not only $n, \bar{\ell}, \lambda$, but also $\tau_0, \tau_1, \cdots, \tau_{\bar{\ell}-1}$ as well. Therefore, the parameters of a ZDB function shall not be considered new only when all of the parameters $\{n,\{\tau_0,\ldots,\tau_{\bar{\ell}-1}\},\bar{\ell},\lambda\}$ of this ZDB function can be obtained by an earlier constructed ZDB function.

\section*{Acknowledgments}


Cunsheng Ding's and Maosheng Xiong's researches are supported by the Hong Kong Research Grants Council under Grant (Nos. 600812, Nos. 606211 and SBI12SC05), respectively.



\begin{thebibliography}{10}

\bibitem{Cai13}
H.~Cai, X.~Zeng, T.~Helleseth, X.~Tang, and Y.~Yang, ``A new construction of
  zero-difference balanced functions and its applications,'' {\em IEEE Trans.
  Inform. Theory}, vol.~59, no.~8, pp.~5008--5015, 2013.

\bibitem{Ding08}
C.~Ding, ``Optimal constant composition codes from zero-difference balanced
  functions,'' {\em IEEE Trans. Inform. Theory}, vol.~54, no.~12,
  pp.~5766--5770, 2008.

\bibitem{Ding09}
C.~Ding, ``Optimal and perfect difference systems of sets,'' {\em J. Combin.
  Theory Ser. A}, vol.~116, no.~1, pp.~109--119, 2009.

\bibitem{DT12}
C.~Ding and Y.~Tan, ``Zero-difference balanced functions with applications,''
  {\em Journal of Statistical Theory and Practice}, vol.~6, no.~1, pp.~3--19,
  2012.

\bibitem{DY051}
C.~Ding and J.~Yin, ``Algebraic constructions of constant composition codes,''
  {\em IEEE Trans. Inform. Theory}, vol.~51, no.~4, pp.~1585--1589, 2005.

\bibitem{DY05}
C.~Ding and J.~Yin, ``Combinatorial constructions of optimal
  constant-composition codes,'' {\em IEEE Trans. Inform. Theory}, vol.~51,
  no.~10, pp.~3671--3674, 2005.

\bibitem{Yuan05}
C.~Ding and J.~Yuan, ``A family of optimal constant-composition codes,'' {\em
  IEEE Trans. Inform. Theory}, vol.~51, no.~10, pp.~3668--3671, 2005.
  
\bibitem{Feng} T. Feng, ``A new construction of perfect nonlinear functions using 
Galois rings," {\em J. Comb. Designs,} vol. 17, no. 3, pp. 229--239, April 2009.   

\bibitem{GFM06}
G.~Ge, R.~Fuji-Hara, and Y.~Miao, ``Further combinatorial constructions for
  optimal frequency-hopping sequences,'' {\em J. Combin. Theory Ser. A},
  vol.~113, no.~8, pp.~1699--1718, 2006.

\bibitem{GMY09}
G.~Ge, Y.~Miao, and Z.~Yao, ``Optimal frequency hopping sequences: auto- and
  cross-correlation properties,'' {\em IEEE Trans. Inform. Theory}, vol.~55,
  no.~2, pp.~867--879, 2009.
  
\bibitem{Hou} X.-D. Hou, ``Cubic bent functions," {\em Discrete Mathematics,} 
vol. 189, nos. 1--3, pp. 149--161, July 1998.    

\bibitem{Lev71}
V.~I. Leven{\v{s}}te{\u\i}n, ``A certain method of constructing quasilinear
  codes that guarantee synchronization in the presence of errors,'' {\em
  Problemy Pereda\v ci Informacii}, vol.~7, no.~3, pp.~30--40, 1971.

\bibitem{Lev04}
V.~I. Leven{\v{s}}te{\u\i}n, ``Combinatorial problems motivated by comma-free
  codes,'' {\em J. Combin. Des.}, vol.~12, no.~3, pp.~184--196, 2004.


\bibitem{Luo03}
Y.~Luo, F.-W. Fu, A.~J.~H. Vinck, and W.~Chen, ``On constant-composition codes
  over {$Z_q$},'' {\em IEEE Trans. Inform. Theory}, vol.~49, no.~11,
  pp.~3010--3016, 2003.

\bibitem{Ny91}
K.~Nyberg, ``Perfect nonlinear {S}-boxes,'' in {\em Advances in
  cryptology---{EUROCRYPT} '91 ({B}righton, 1991)}, vol.~547 of {\em Lecture
  Notes in Comput. Sci.}, pp.~378--386, Berlin: Springer, 1991.

\bibitem{PW13}
A.~Pott and Q.~Wang, ``Difference balanced functions and their generalized
  difference sets,'' {\em arXiv preprint arXiv:1309.7842}, 2013.

\bibitem{Wang06}
H.~Wang, ``A new bound for difference systems of sets,'' {\em J. Combin. Math.
  Combin. Comput.}, vol.~58, pp.~161--167, 2006.

\bibitem{WZ12}
Q.~Wang and Y.~Zhou, ``Sets of zero-difference balanced functions and their
  applications,'' {\em arXiv preprint arXiv:1208.1878}, 2012. 
  
\bibitem{ZGY} X. Zeng, H. Guo, and J. Yuan, ``A note of perfect nonlinear functions," 
in {\em Cryptography and Network Security},  vol.~4301 of {\em Lecture
  Notes in Comput. Sci.}, pp.~259--269, Berlin: Springer, 2006.  

\bibitem{ZKW}   Z. Zha, G. M. Kyureghyan, X. Wang, ``Perfect nonlinear binomials and their 
semifields," {\em Finite Fields Appl.}, vol. 15, no. 2, pp. 125--133, April 2009. 


\bibitem{Zhou12}
Z.~Zhou, X.~Tang, D.~Wu, and Y.~Yang, ``Some new classes of zero-difference
  balanced functions,'' {\em IEEE Trans. Inform. Theory}, vol.~58, no.~1,
  pp.~139--145, 2012. 
  
\end{thebibliography}

\end{document}